\documentclass[a4paper,english]{lipics-v2021}
\nolinenumbers

\usepackage{thmtools}

\usepackage{booktabs}
\usepackage{xspace}
\usepackage{amstext,amssymb}
\usepackage{url}
\usepackage{stmaryrd}
\usepackage{apxproof}
\usepackage{enumerate}

\usepackage[dvipsnames]{xcolor}

\theoremstyle{plain}
\newtheoremrep{theorem}{Theorem}[section]
\newtheoremrep{lemma}[theorem]{Lemma}
\newtheoremrep{claim}[theorem]{Claim}
\newtheoremrep{corollary}[theorem]{Corollary}
\newtheoremrep{proposition}[theorem]{Proposition}
\newtheoremrep{conjecture}[theorem]{Conjecture}

\usepackage{tikz}
\usetikzlibrary{calc}

\usepackage{algpseudocode,algorithm}

\newcommand{\nc}{\newcommand}

\nc{\nat}{\ensuremath{\mathbb{N}}}
\nc{\V}{\ensuremath{\mathsf{Val}}}
\nc{\values}{\V}
\nc{\names}{\ensuremath{\mathsf{Names}}\xspace}
\nc{\attnames}{\ensuremath{\mathsf{Att}}}
\nc{\unary}{\ensuremath{\mathsf{Unary}}}
\nc{\len}{\ensuremath{\mathsf{len}}\xspace}

\nc{\SigmaCF}{\ensuremath{\Sigma_\text{CF}}\xspace}
\nc{\SigmaFR}{\ensuremath{\Sigma_\text{FR}}\xspace}

\nc{\arity}{\text{ar}\xspace}
\nc{\dom}{\text{Dom}\xspace}

\nc{\cR}{{\mathcal R}}
\nc{\lab}{\text{lab}}
\nc{\nfa}{\text{NFA}\xspace}
\nc{\product}{\text{product}\xspace}
\nc{\productautomaton}{\product}

\nc{\np}{NP\xspace}
\nc{\ptime}{PTIME\xspace}
\nc{\sharpp}{\#P\xspace}

\nc{\rpq}{\text{RPQ}\xspace}
\nc{\rpqs}{\text{RPQs}\xspace}

\nc{\sem}[1]{\ensuremath{\llbracket #1 \rrbracket}}

\theoremstyle{plain}

\nc{\todo}[1]{{\color{red}\textbf{TODO:} #1}}

\newcommand{\pspace}[1]{\textsc{pspace}\xspace}

\newcommand{\eat}[1]{}
\def\ar{\mathord{\mathit{ar}}}

\newcommand{\defeq}{\mathrel{\mathop:}=}

\def\word{\mathrm{word}}

\title{A Formal Language Perspective on Factorized Representations}

\author{Benny Kimelfeld}
{Technion, Haifa, Israel}
{bennyk@technion.ac.il}
{https://orcid.org/0000-0002-7156-1572}
{}

\author{Wim Martens}
{University of Bayreuth, Bayreuth, Germany}
{wim.martens@uni-bayreuth.de}
{https://orcid.org/0000-0001-9480-3522}
{Supported by ANR project EQUUS ANR-19-CE48-0019; funded by the Deutsche Forschungsgemeinschaft (DFG, German Research Foundation), project number 431183758.}

\author{Matthias Niewerth}
{University of Bayreuth, Bayreuth, Germany}
{matthias.niewerth@uni-bayreuth.de}
{https://orcid.org/0000-0003-2032-5374}
{}

\begin{document}

\authorrunning{Benny Kimelfeld, Wim Martens, and Matthias Niewerth}
\Copyright{Benny Kimelfeld, Wim Martens, and Matthias Niewerth}



\maketitle
\begin{abstract}
  Factorized representations (FRs) are a well-known tool to succinctly represent
  results of join queries and have been originally defined using the named
  database perspective. We define FRs in the unnamed database perspective and
  use them to establish several new connections. First, unnamed FRs can be
  exponentially more succinct than named FRs, but this difference can be
  alleviated by imposing a disjointness condition on columns. Conversely, named
  FRs can also be exponentially more succinct than unnamed FRs. Second, unnamed FRs are the
  same as (i.e., isomorphic to) context-free grammars for languages in which
  each word has the same length. This tight connection allows us to transfer a
  wide range of results on context-free grammars to database factorization; of
  which we offer a selection in the paper. Third, when we
  generalize unnamed FRs to arbitrary sets of tuples, they become a
  generalization of \emph{path multiset representations}, a formalism that was
  recently introduced to succinctly represent sets of paths in the context of
  graph database query evaluation.
\end{abstract}

\section{Introduction}

Factorized databases (FDBs) aim at succinctly representing the result of join
queries by systematically avoiding
redundancy. Since their introduction by Olteanu and Zavodny~\cite{OlteanuZ-icdt12,OlteanuZ-tods15}, they were the
inspiration and key technical approach toward the development of
algorithms for efficient query
evaluation~\cite{BakibayevOZ-vdlb12,OlteanuS-vldb16}, including the
construction of direct-access structures for join
queries~\cite{DBLP:journals/tods/CarmeliTGKR23,DBLP:journals/tods/CarmeliZBCKS22},
evaluation of aggregate
queries~\cite{BakibayevKOZ-vldb13,DBLP:conf/pods/Torunczyk20,NikolicO-sigmod18},
and the application of machine-learning
algorithms over databases~\cite{SchleichOC-sigmod16}.

At the core of factorized databases are factorized relations. In essence,
a \emph{factorized relation (FR)} is a relational algebra query that builds the represented
set of tuples. It involves data values and only two operators: \emph{union}
and \emph{Cartesian product}. The restriction to these two operators
provides succinctness and, at the same time, ensures the efficiency of
downstream operations. 
We refer to \cite{OlteanuS-SR}
for a gentle introduction into factorized databases.

Factorized relations have been introduced in the \emph{named database
  perspective}, where tuples are defined as functions from a set of
attribute names to a set of attribute values \cite{ABLMP21} and are therefore unordered. In
this paper, we explore FRs from the \emph{unnamed perspective}, where tuples are
ordered. Our motivation to explore this perspective is twofold. First, we
believe that tuples in many database systems are ordered objects and second, we
want to understand the relationship between factorized relations and
\emph{path multiset representations (PMRs)}.

\emph{Path Multiset Representations (PMRs)}
were recently introduced to succinctly represent \emph{(multi)sets of
  paths} in graph databases \cite{MartensNPRVV-vldb23}. In particular,
they aim at representing paths that match \emph{regular path queries}
(RPQs), which are the fundamental building block of modern graph
database pattern
matching~\cite{DeutschFGHLLLMM-sigmod22,FrancisGGLMMMPR-icdt23} and
have been studied for decades
\cite{BarceloFR-icalp19,BarceloHLW-pods10,BarceloLR-pods11,CalvaneseGLV-pods99,CalvaneseGLV-kr00,ConsensM-pods90,CruzMW-sigmod87,FigueiraGKMNT-kr20,Schmid-pods20}.
Compared to traditional research, modern graph query languages such as
Cypher~\cite{cypher,FrancisGGLLMPRS-sigmod18}, SQL/PGQ
\cite{sql-pgq,DeutschFGHLLLMM-sigmod22}, and GQL
\cite{gql,DeutschFGHLLLMM-sigmod22} use RPQs in a fundamentally new
way. In a nutshell, in most of the research literature, an RPQ $q$
returns pairs of endpoints of paths that are matched by $q$. In
Cypher, SQL/PGQ, and GQL, it is possible to return the actual paths
that match $q$~\cite{DeutschFGHLLLMM-sigmod22,FrancisGGLMMMPR-icdt23},
which is much less
explored~\cite{MartensNT-stacs20,MartensP-pods22,MartensT-icdt18}. The
challenge for PMRs is to succinctly represent the (possibly
exponentially many or even infinitely many) paths that match an RPQ,
and to allow query operations to be performed directly on the PMR.  In fact,
several experimental studies show that using PMRs can drastically speed up
query evaluation~\cite{MartensNPRVV-vldb23,pathfinder,pathfinderdemo}.

While FDBs and PMRs aim at the same purpose of succinctly representing
large (intermediate) results of queries, they are quite different. For
instance, factorized databases represent database relations, which are
tuples of the same length and which are always finite. Path multiset
representations represent (multi)sets of paths in graphs, with varying
lengths, and these sets can be infinite. Finally, even
though going from a fixed-length setting to a varying-length setting
seems like a generalization, it is not clear how PMRs
generalize factorized databases. 
Viewing FRs through the named database perspective however, will make the
relationship between FRs and PMRs much clearer.

Our contributions are the following. Let us use the term \emph{named factorized
  representations (nFR)} to refer to the \emph{$d$-representations}, also called
\emph{factorized representations with definitions}, introduced by Olteanu and
Zavodny~\cite{OlteanuZ-icdt12,OlteanuZ-tods15}. We define unnamed factorized
relations (uFRs) which are, analogously to the named case, relational algebra
expressions built from data values, union, and Cartesian product. Although uFRs
are conceptually very similar to nFRs, they are incomparable in size, since
worst-case exponential blow-ups exist in both directions. The blow-up from uFRs
to nFRs disappears, however, when we impose a disjointness condition on the
columns of uFRs. We then observe that there exists a bijection $\beta$ between
uFRs and a class of context-free grammars (CFGs) for languages in which all
words have the same length. Furthermore, each uFR $F$ is isomorphic to the
grammar $\beta(F)$ and its represented relation is the straightforward encoding
of the language of $\beta(F)$ as tuples. Loosely speaking, this means that uFRs
and this class of context-free grammars are the same thing.

This tight connection between uFRs and CFGs allows us to immediately infer a
number of complexity results on uFRs, e.g., on their membership problem, on their
equivalence problem (``Do two uFRs represent the same set of tuples?''), on the
counting problem (``How many different tuples are represented by a uFR?''),
their enumeration problem, and
on size lower bounds. It also allows us to generalize uFRs to a setting in which
database relations can contain tuples of different arities, which is a model
that is currently being implemented by RelationalAI~\cite{rel}.

Finally, it allows us to clarify the connection between uFRs and PMRs. Loosely
speaking, whereas uFRs are context-free grammars for uniform-length languages,
PMRs are non-deterministic finite automata. In this sense, PMRs are
 more expressive than uFRs, because they can represent infinite
objects. On uniform-length languages, however, uFRs are more general: PMRs and
uFRs can represent the same languages, but PMRs are a special case of uFRs.
Indeed, it is well-known that CFGs with rules only of the form $A \to bC$ and $A
\to b$ are isomorphic to non-deterministic finite automata.

Once we understand how PMRs and uFRs compare, we ask ourselves what are the
tradeoffs between them if one would use them as compact representations in the
same system. In principle, a CFG for a finite language can be converted into an
NFA or, conversely, an NFA representing a finite number of paths could be made
even more succinct as a CFG. Different representations may have different
benefits. For example, if we want to compute the complement of a relation $R$
represented by a uFR (i.e., a CFG), it may be a viable plan to convert the CFG
into a DFA and use the trivial complement operation on DFAs (slightly modified so
that we only take the complement on words of the relevant length). As such, even
a naive algorithm for complementation is able to avoid fully materializing
$R$ in some cases (for example, the cases where the CFG is right-linear). In
this paper, we embark on this (quite extensive) tradeoff question and
investigate relative blow-ups between the different relevant classes of
context-free grammars and automata.

\noindent \textbf{Further Related Work.} The present paper mainly aims at
connecting areas that were previously thought to be different (to the best of
our knowledge). In the named database perspective, factorized relations are
known to be closely connected to decision diagrams. (See,
e.g.,~\cite{AmarilliBJM17}.) A recent overview on connections between binary
decision diagrams (BDDs) and various kinds of automata was made by Amarilli et
al.~\cite{circuits}. The authors focus on translations that preserve properties
of interest, like the number of objects represented (variables and truth
assignments), even if the objects themselves change in the translation.
Here, on the other hand, we focus on much stronger correspondences, namely
isomorphisms. Put differently, we are interested in preserving not only
certain properties of the represented objects, but even the objects themselves.
Another
difference between our work and the research on circuits is that 
the latter usually represent
\emph{unordered} objects; we focus on \emph{ordered} tuples and paths (and
connect with factorizations for unordered tuples in Section~\ref{sec:named-unnamed}).

This paper is certainly not the first that considers context-free
grammars for compression purposes. In fact, the area of \emph{straight-line programs}, see,
e.g.,~\cite{CharikarLLPPSS05,GanardiJL21,Lohrey12} is completely centered around
this idea. The focus there, however, is on compactly representing a single word.
Factorized relations on the other hand focus on representing a database relation; or a
finite set of words.


\section{Factorized Relations}\label{sec:preliminaries}\label{sec:definitions}

Factorized relations were introduced by Olteanu and
Z\'avodn\'y~\cite{OlteanuZ-tods15}, who defined them in the named database
perspective, where tuples are unordered~\cite{ABLMP21}. We investigate them in
the unnamed perspective (i.e., for ordered tuples).
Appendix~\ref{app:named} contains a brief recap on the named and unnamed perspective.

Let $\values$ be a countably infinite set of \emph{values}.
A \emph{($k$-ary) database tuple} is an element $\vec a$ of $\V^k$. For every
$i=1,\dots,k$, we denote by $\vec a[i]$ the $i$th value of $\vec a$. A
\emph{database relation} is a finite set $R$ of tuples of the same arity. Let
$\names$ be a countably infinite set of \emph{names}, disjoint from $\values$.

\subsection{Unnamed Factorized Relations}
For defining \emph{unnamed factorized relations (uFR)}, we
follow the intuition that factorized relations are relational algebra
expressions that can use names to re-use subexpressions \cite{OlteanuZ-tods15}.
We also follow \cite{OlteanuZ-tods15}
in disallowing unions of $\emptyset$ with non-empty relations. Let
$X \subseteq \names$ be a finite set of \emph{expression names} and let $\ar
\colon X \to \nat$ be a function that associates an \emph{arity} to each name in
$X$. A \emph{relation expression that references $X$}, or \emph{$X$-expression}
for short, is a relational algebra expression built from singletons, products,
unions, and names from $X$. We inductively define $X$-expressions $E$ and their
associated arity $\ar(E)$ as follows:
\begin{description}
\item[(empty)] $E=\emptyset$ is an $X$-expression with $\ar(E) = 0$;
\item[(nullary tuple)] $E=\langle \rangle$ is an $X$-expression with  $\ar(E)=0$;
\item[(singleton)] $E=\langle a \rangle$ is an $X$-expression for each $a \in \V$, with
  $\ar(E)=1$;
\item[(name reference)] $E=N$ is an $X$-expression for each $N\in X$,  with $\ar(E) =
  \ar(N)$;
\item[(union)] for $X$-expressions $E_1, \ldots, E_n$ with $\ar(E_1) = \cdots = \ar(E_n)$, we have that $E=(E_1 \cup
  \cdots \cup E_n)$ is an $X$-expression with
   $\ar(E) = \ar(E_1) = \cdots = \ar(E_n)$;
 \item[(product)] for $X$-expressions $E_1, \ldots, E_n$, we have that $E=(E_1 \times \cdots
  \times E_n)$ is an $X$-expression with $\ar(E)=\sum_{i\in[n]}\ar(E_i)$.
\end{description}
\begin{definition}\label{def:uFR}
A \emph{$k$-ary unnamed factorized relation (uFR)} is a pair $(S,D)$, where $S \in \names$ is
the start symbol and $D = \{N_1 \defeq E_1, \ldots , N_n \defeq E_n\}$ is a set
of expressions where:
\begin{enumerate}
\item $N_1 = S$ and $\ar(N_1 = k)$;
\item each $N_i$ is an expression name;
\item each $E_i$ is an $X_i$-expression for $X_i=\{N_{i+1},\dots,N_{n}\}$; and
\item $\ar(N_{i})=\ar(E_i)$ for all $i=1,\dots,n$.
\end{enumerate}
\end{definition}
Note that the expression $E_n$ does not use name references. Hence,
$E_n$ can be evaluated without resolving references, and the result is
a relation $R_n$. Once we have $R_n$, we can construct $R_{n-1}$ from
$E_n$ by replacing $N_n$ with $R_n$. We continue this way until we
obtain $R_1$, which is a $k$-ary relation, and is the relation that
$F$ represents. We denote this relation, $R_1$, by $\sem{F}$. Furthermore, we
will denote $R_i$ by $\sem{N_i}$ for every $i \in [n]$.

Note that, since $D$ is a set, we should indicate which element is $S$. Indeed,
choosing a different start symbol changes $\sem{F}$ (and may make some parts of
$F$ useless since they cannot be reached from $S$). One may use a notational
convention that we always take $S$ to be the first name that we write down in
$D$, as is done in \cite[Section 4.2]{OlteanuZ-tods15}.

\begin{figure}[t]
  \centering
  \begin{tikzpicture}
    \node [baseline=(current bounding box.north)] at (0,0) {
      \begin{tabular}[t]{c@{\hspace{3mm}}c}
        \multicolumn{2}{l}{\texttt{Customer}}\\
        \toprule
        \texttt{cid} & \texttt{name}\\
        \midrule
        c1 & n1 \\
        c2 & n2\\
        c3 & n3\\
        \bottomrule
      \end{tabular}
    };

    \node [baseline=(current bounding box.north)] at (3.2,-.68) {
      \begin{tabular}[t]{cc}
        \multicolumn{2}{l}{\texttt{PurchaseHistory}}\\
        \toprule
        \texttt{cid} & \texttt{product}\\
        \midrule
        c1 & flute\\
        c1 & wire\\
        c2 & harp\\
        c2 & wire\\
        c3 & flute\\
        c3 & phone\\
        \bottomrule
      \end{tabular}
    };

    \node [baseline=(current bounding box.north)] at (7,-.68) {
      \begin{tabular}[t]{c@{\hspace{4mm}}c}
        \multicolumn{2}{l}{\texttt{Supply}}\\
        \toprule
        \texttt{product} & \texttt{supplier}\\
        \midrule
        flute & s1 \\
        flute & s2 \\
        wire & s3\\
        harp & s1 \\
        harp & s2 \\
        phone & s3 \\
        \bottomrule
      \end{tabular}
    };
  \end{tikzpicture}
  \caption{A database.}
  \label{fig:database}
\end{figure}

\begin{example}\label{ex:FR}
  Consider the database in Figure~\ref{fig:database}. We assume that the
  attributes of each relation are ordered from left to right, i.e., \texttt{cid}
  is the first attribute of \texttt{Customer}, and so on.
  The following is a
  factorized relation that represents \texttt{Customer} $\bowtie$
  \texttt{PurchaseHistory} $\bowtie$ \texttt{Supply}. (The ordering of rules
  is top-to-bottom, left-to-right; and the start symbol is $N_1$.)
  \begin{center}
    $
    \begin{array}{ll@{\hspace{1cm}}ll@{\hspace{1cm}}ll}
      N_1  & := A_1\cup A_2\cup A_3  &   B_2 & := \text{c2}\times C_2 & P_1  & := \text{flute}\times D \\
      A_1  & := B_1\times \text{n1}   &   B_3 & := \text{c3}\times C_3 & P_2 &:= \text{wire}\times \text{s3}\\
      A_2  & := B_2\times \text{n2}   &   C_1  & := P_1\cup P_2 & P_3  & := \text{harp}\times D \\
      A_3  & := B_3\times \text{n3}   &   C_2 & := P_2\cup P_3 & P_4 &:= \text{phone}\times \text{s3}\\
      B_1 & := \text{c1}\times C_1    &   C_3 & := P_1\cup P_4 & D & := \text{s1} \cup \text{s2}\\
    \end{array}
    $
  \end{center}
  Figure~\ref{fig:factorized} contains a visualization of the factorized
  relation. (We annotated some of the elements of \names in orange and omitted
  $\langle \cdot \rangle$ around data values.)
\end{example}

\begin{figure}
  \centering
  \begin{subcaptionblock}{.5\textwidth}
    \begin{tikzpicture}[level distance=9mm]
      \node at (0,0) {$\cup$}
      [sibling distance=2.5cm]
      child {node (A1) {$\times$}
        [sibling distance=5mm]
        child {node {} edge from parent[draw=none]}
        child {node (B1) {$\times$}
          [sibling distance=8mm]
          child {node {c1}}
          child {node (union1) {$\cup$}
            child {node (timesleft) {$\times$}
              child {node {flute}}
              child {node (unionleft) {$\cup$}
                child {node {s1}}
                child {node {s2}}
              }
            }
            child {node {} edge from parent[draw=none]}
          }
        }
        child[level distance=6mm] {node {n1}}
      }
      child {node (A2) {$\times$}
        [sibling distance=5mm]
        child {node {} edge from parent[draw=none]}
        child {node (B2) {$\times$}
          [sibling distance=8mm]
          child {node {c2}}
          child {node (union2) {$\cup$}
            child {node {} edge from parent[draw=none]}
            child {node (timesmiddle) {$\times$}
              child {node {harp}}
              child {node {} edge from parent[draw=none]}
            }
          }
        }
        child[level distance=6mm] {node {n2}}
      }
      child {node (A3) {$\times$}
        [sibling distance=5mm]
        child {node {} edge from parent[draw=none]}
        child {node (B3) {$\times$}
          [sibling distance=8mm]
          child {node {c3}}
          child {node (union3) {$\cup$}
            child {node (timesright) {$\times$}
              child {node {phone}}
              child {node {s3}}
            }
          }
        }
        child[level distance=6mm] {node {n3}}
      }
      ;

      \node (timesroot) at ($(timesleft)+(1.8,0)$) {$\times$}
      [sibling distance=6mm]
      child {node {wire}}
      child {node {s3}}
      ;

      \path
      (union1) edge (timesroot)
      (union2) edge (timesroot)
      (union3) edge[out=210,in=30] (timesleft)
      (timesmiddle) edge[out=-50,in=55] (unionleft)
      ;

      \node at ($(A1)+(-.4,.2)$) {\small \textcolor{Orange}{$A_1$}};
      \node at ($(A2)+(-.4,.2)$) {\small \textcolor{Orange}{$A_2$}};
      \node at ($(A3)+(.4,.2)$) {\small \textcolor{Orange}{$A_3$}};
      \node at ($(B1)+(-.4,.2)$) {\small \textcolor{Orange}{$B_1$}};
      \node at ($(B2)+(-.4,.2)$) {\small \textcolor{Orange}{$B_2$}};
      \node at ($(B3)+(-.4,.2)$) {\small \textcolor{Orange}{$B_3$}};
      \node at ($(union1)+(.4,.2)$) {\small \textcolor{Orange}{$C_1$}};
      \node at ($(union2)+(.4,.2)$) {\small \textcolor{Orange}{$C_2$}};
      \node at ($(union3)+(.4,.2)$) {\small \textcolor{Orange}{$C_3$}};
      \node at ($(timesleft)+(-.4,.2)$) {\small \textcolor{Orange}{$P_1$}};
      \node at ($(timesroot)+(-.4,0)$) {\small \textcolor{Orange}{$P_2$}};
      \node at ($(timesmiddle)+(.4,0)$) {\small \textcolor{Orange}{$P_3$}};
      \node at ($(timesright)+(.4,0)$) {\small \textcolor{Orange}{$P_4$}};
      \node at ($(unionleft)+(.3,.1)$) {\small \textcolor{Orange}{$D$}};
    \end{tikzpicture}
    \caption{Visualization of a factorized relation representing \texttt{Customer} $\bowtie$
      \texttt{PurchaseHistory} $\bowtie$ \texttt{Supplies}\\ \label{fig:factorized}}
  \end{subcaptionblock}\quad
  \begin{subcaptionblock}{.4\textwidth}
    \centering
    \def\val#1{\mbox{#1}}
    \parbox{2in}{
      \begin{align*}
        &  S  \to A_1\cup A_2\cup A_3 \\
        &A_1  \to B_1\cdot \val{n1}\quad\quad A_2  \to B_2\cdot \val{n2}\\
        &A_3  \to B_3\cdot \val{n3} \\
        &B_1  \to \val{c1}\cdot C_1 \quad\quad B_2 \to  \val{c2}\cdot C_2\\
        &B_3  \to \val{c3}\cdot C_3 \\
        &C_1  \to P_1\cup P_2 \quad\quad C_2 \to P_2\cup P_3\\
        &C_3 \to P_1\cup P_4 \\
 &P_1  \to \val{flute}\cdot D \quad\quad  P_2 \to \val{wire}\cdot \val{s3}\\
 &P_3  \to \val{harp}\cdot D \quad\quad   P_4 \to \val{phone}\cdot \val{s3}\\
 &D \to \val{s1} \cup \val{s2}
\end{align*}
      }
    \caption{Context-free grammar isomorphic to the factorized relation\label{fig:CFG}}
  \end{subcaptionblock}

  \caption{Factorized relation and context-free grammar representing the join of
    the relations in Figure~\ref{fig:database}\label{fig:isomorphism}}
\end{figure}

\subsection{Relationship to Named Factorized Relations}\label{sec:named-unnamed}
We refer to the \emph{$d$-representations} of Olteanu and Zavodny
\cite{OlteanuZ-tods15} as \emph{named factorized relations (nFR)}. 
For an nFR $F$, we use $\sem{F}$ to denote the named
relation $R$ represented by $F$.\footnote{$\sem{F}$ is formally defined in
  \cite{OlteanuZ-tods15}. The definition is completely analogous to
  Definition~\ref{def:uFR} but starts from tuples in the unnamed perspective. We
  therefore do not repeat the definition here.}
Although the definitions of named and unnamed factorized relations are 
similar, we show that unnamed factorized relations can be exponentially more
succinct than named factorized relations and vice versa. In the direction from
uFRs to nFRs, the exponential size difference is due to the capability of uFRs to be
able to factorize `horizontally' and can be alleviated by imposing that every
column uses different values. The exponential size difference in the other
direction is due to the unordered nature of nFRs and tuples in the named
perspective.

In order to compare uFRs with nFRs, we need to explain when we consider an nFR
and a uFR to be equivalent. We consider the standard conversion between named
and unnamed database relations in \cite[Chapter 2]{ABLMP21} and
write $\textsf{unnamed}(R)$ for the unnamed relation obtained from converting
the named relation $R$ to the unnamed perspective. (Intuitively, the conversion
assumes a fixed ordering on relation-attribute pairs and converts tuples
$\langle A_1\colon a_1, \ldots, A_k\colon a_k \rangle$ in a named relation $R$
into tuples $(a_1,\ldots,a_k)$.) Finally, we say that an nFR $F_1$ and uFR $F_2$
are \emph{equivalent} if $\textsf{unnamed}(\sem{F_1}) = \sem{F_2}$.

\begin{proposition}
  \begin{enumerate}[(1)]
  \item For each $n \in \nat$, there exists a uFR of size $O(n)$ such that the smallest
    equivalent nFR is exponentially larger.
  \item For each $n \in \nat$, there exists an nFR of size $O(n)$ such that the
    smallest equivalent uFR is exponentially larger.
  \end{enumerate}
\end{proposition}
\begin{proof}[Proof sketch]
  For (1), consider the unnamed factorized relation $F$ with the expressions
  \[N_1 := N_{2} \times N_{2} \qquad N_2 := N_3 \times N_3 \qquad \cdots \qquad
    N_{n-1} := N_n \times N_n \qquad N_n := \langle 0 \rangle \cup \langle 1
    \rangle \;.\] This uFR defines the set of all $2^n$-ary tuples with values in
  $\{0,1\}$. Notice that the uFR is exponentially smaller than the arities of the
  tuples that it defines, and doubly exponentially smaller than the number of
  tuples in $\sem{F}$, which is $2^{2^{n}}$. The following can be shown using a
  standard inductive argument.
  \begin{claim}\label{claim:nFR-size}
    For every nFR $F$, the size of $F$ is at most exponentially smaller than
    the size of $\sem{F}$.
  \end{claim}
  
  For (2), consider a relation $R[A_1,\ldots,A_n,B_1,\ldots,B_n]$ such that
  $R.A_1 < \cdots R.A_n < R.B_1 < \cdots < R.B_n$ in the attribite ordering for
  converting from the named to unnamed perspective. Consider the set of tuples
  $R = \{t \mid \pi_{A_1,\ldots,A_n} t = \pi_{B_1,\ldots,B_n} t \}$ in which all
  data values are $0$ or $1$. It is easy
  to construct an nFR for $R$ of size $O(n)$: take $S_1 := ((\langle A_1:0 \rangle \times
  \langle B_1:0 \rangle)) \cup ((\langle A_1:1 \rangle \times \langle B_1:1 \rangle))
  \times S_2$, $S_2 := ((\langle A_2:0 \rangle \times
  \langle B_2:0 \rangle)) \cup ((\langle A_2:1 \rangle \times \langle B_2:1 \rangle))
  \times S_3$, etc.
  The lower
  bound for uFRs is proved in Corollary~\ref{cor:lowerbound}.
\end{proof}

The significant size increase when going from the unnamed to the named
perspective exists because tuples in the named perspective are ``never
same in different columns''. Indeed, in the named perspective, tuples $\langle
A_1\colon a_1, \ldots, A_k\colon a_k \rangle$ are such that all attribute names
$A_i$ are pairwise distinct. Therefore, the composite values $A_i\colon a_i$ are
also pairwise distinct.

In the unnamed perspective, we can enforce the same restriction by requiring
that coordinates in tuples come from disjoint domains. (This notion is similar
to the notion of \emph{decomposability} in circuits \cite{circuits}.) We say that a relation $R$
\emph{has disjoint positions} if no value occurs in two different columns. More
precisely, if $\vec a[i]\neq\vec b[j]$ for all $\vec a,\vec b\in R$ and $1\leq
i<j\leq k$. We say that an unnamed factorized relation $F$ \emph{has disjoint
  positions} if $\sem{F}$ has disjoint positions.

\begin{toappendix}
Claim~\ref{claim:nFR-size} follows from item (2) in the following proposition.  
\begin{observation}\label{obs:disjoint-positions-blowup}
  \begin{enumerate}
  \item Let $F$ be an nFR. Then $|\sem{F}| \leq 2^{p(|F|)}$ for some polynomial $p$.
  \item Let $F$ be an uFR with disjoint positions. Then $|\sem{F}| \leq 2^{p(|F|)}$
    for some polynomial $p$.

  \end{enumerate}
\end{observation}
\begin{proof}
  We first prove (2).
  If $F$ does not encode the empty relation, it has at least one symbol for each
  attribute. As $F$ has disjoint positions, we can conclude that the arity is
  bounded by $|F|$. The number of tuples in a relation is bounded by the number
  of symbols to the power of the arity. Therefore $|\sem{F}| \leq |F|^{|F|} \in
  2^{{\mathcal O} (|F| \log |F|)}$.

  The proof of (1) is analogous.
\end{proof}
\end{toappendix}

\begin{propositionrep}
  For each $n \in \nat$ and uFR of size $n$ with disjoint positions, there is an
  equivalent nFR of size $n$.
\end{propositionrep}
\begin{proof}[Proof sketch]
  If the uFR has disjoint positions, we can simply assign a fixed attribute name
  to each ``position'' in tuples. The construction is then straightforward
  (union remains union; and product remains product). Correctness follows from a
  simple inductive argument.
\end{proof}


\section{Context-Free Grammars and Their Connection to FRs}

Let $\Sigma$ be a finite set, whose elements we call \emph{symbols}. By
$\varepsilon$ we denote the empty word, that is, the word of length $0$. By
$\Sigma^*$ we denote the set of all words over $\Sigma$, i.e., the set of words
$w = a_1 \cdots a_n$, where $\emptyset$ are not elements of $\Sigma$. A
\emph{regular expression (RE)} over $\Sigma$ is inductively defined as follows.
Every $a \in \Sigma$ is a regular expression, and so are the symbols
$\varepsilon$ and $\emptyset$. Furthermore, if $e_1$ and $e_2$ are regular
expressions over $\Sigma$, then so are $e_1 \cdot e_2$ (concatenation), $e_1
\cup e_2$ (union), and $e_1^*$ (Kleene star). As usual, we often omit the
concatenation operator in our notation. When taking $e_1 \cup e_2$, we assume
that $e_1$ nor $e_2$ are $\emptyset$.\footnote{Indeed, a union with $\emptyset$
  is never useful and the restriction is easy to enforce. We make this
  restriction because unions in uFRs are defined with the same restriction.
  Theorem~\ref{theo:FR-ECFG-uniform}, which states that uFRs and uniform-length ECFGs are the same,
  also holds if unions in uFRs and REs both allow the empty set.} By $L(e)$ we
define the language of $e$, which is defined as usual.

\begin{definition}[see, e.g., \cite{MadsenK-acta76}]
An \emph{extended context-free grammar} (abbreviated \emph{ECFG}) consists of rules where
nonterminals are defined using arbitrary regular expressions over terminals and
nonterminals. Formally, an ECFG is a tuple $G = (T,N,S,R)$ where:
\begin{itemize}
\item $T$ is a finite set of \emph{terminals};
\item $N$ is a finite set of \emph{nonterminals} such that $T \cap N =
  \emptyset$;
\item $S \in N$ is the \emph{start symbol}; and
\item $R$ is a finite set of \emph{rules} of the form $A \rightarrow e$, where
  $e$ is a regular expression over $T\cup N$.
\end{itemize}
For the purpose of this paper, we will always choose $T \subseteq \V$ and $N
\subseteq \names$. Furthermore, we assume that all terminals in $T$ are actually
used in the grammar, that is, for each terminal $a$, there exists a rule $A
\rightarrow e$ such that $a$ appears in the expression $e$.
\end{definition}

The \emph{language of $G$}, denoted $L(G)$, is defined as usual. A
\emph{derivation step of $G$} is a pair $(u,v)$ of words in $(T \cup N)^*$ such
that $u = \alpha X \beta$ and $v = \alpha \gamma \beta$ where $X \in N$ and
$\alpha,\beta$, $\gamma \in (T \cup N)^*$, and where $R$ contains a rule $X \to
e$ with $\gamma \in L(e)$. We denote such a derivation step as $u \Rightarrow_G
v$. A \emph{derivation} is a sequence $u_0,\ldots,u_n$ such that $u_{i-1}
\Rightarrow_G u_i$ for every $i \in [n]$. We denote by $u \Rightarrow^* v$ that
there exists a derivation that starts in $u$ and ends in $v$ and by $u
\Rightarrow^+ v$ the case where this derivation has at least one step. By $L(u)$ we
denote the language $\{w \in T^* \mid u \Rightarrow_G w\}$. Finally, the
\emph{language of $G$}, denoted $L(G)$, is
the language $L(S)$ of words derived from the start symbol $S$.

A nonterminal $A \in N$ is \emph{useful} if there exists a derivation $S
\Rightarrow^*_G \alpha A \beta \Rightarrow^*_G w$ for some word $w \in
\Sigma^*$. The grammar $G = (T,N,S,R)$ is \emph{trimmed} if every nonterminal in
$N$ is useful. It is well-known that an ECFG can be converted into a trimmed ECFG
in linear time.
A grammar $G$ is \emph{recursive} if there exists a derivation
$\alpha A \beta \Rightarrow^+_G \alpha' A \beta'$ for some $A \in N$ and
$\alpha,\beta,\alpha',\beta' \in (T\cup N)^*$.

\begin{remark}
  \emph{Context-free grammars (CFGs)} are defined analogously to ECFGs, except
  that rules are required to be of the form $A \rightarrow \alpha_1 \cup \cdots
  \cup \alpha_n$,
  where each $\alpha_i$ is a concatenation over $T \cup N$. It
  is well known that CFGs and ECFGs have the same expressiveness and that they
  can be translated back and forth in linear time \cite[p.~202]{HMU-automata}.
\end{remark}

\subsection{Isomorphisms}\label{sec:isomorphism}
In this section we want to define when we consider a uFR and an ECFG to be
isomorphic. To warm up, we first explain when we consider two ECFGs $G_1$ and $G_2$ to be
isomorphic. Intuitively, this is the case when they are the
same up to renaming of non-terminals. Formally, we define isomorphisms using a
function $h \colon \names \to \names$ that we extend to regular expressions as
follows:
\[
\begin{array}{rcll@{\hspace{1cm}}rcl}
  h(\emptyset) &=& \emptyset &&   h(e_1 \cdot e_2) &=& h(e_1) \cdot h(e_2) \\
  h(\varepsilon) &=& \varepsilon &&   h(e_1 \cup e_2) &=& h(e_1) \cup h(e_2) \\
  h(a) &= & a & \text{ for every $a$ in \values} & h(e^*) &=& h(e)^*\\
\end{array}
\]
Then, $G_1 = (T_1,N_1,S_1,R_1)$ is \emph{isomorphic} to $G_2 = (T_2,N_2,S_2,R_2)$ if
there is a bijective function $h \colon N_1 \to N_2$ such that $h(S_1) = S_2$, for
each rule $A \to e$ in $R_1$, the rule $h(A) \to h(e)$ is in
$R_2$, and each rule in $R_2$ is of the form $h(A) \to h(e)$ for some rule $A
\to e$ in $R_1$.

\begin{example}
  (Extended) context-free grammars are isomorphic if and only if they are the
  same up to renaming of non-terminals. For example, the grammars
  \begin{center}
    $\begin{array}{rl@{\hspace{.5cm}}rl}
       S & \to ASB \cup C^* & A & \to a\\
       B & \to b & C & \to c\\
    \end{array}$
    \qquad \qquad and \qquad \qquad
    $\begin{array}{rl@{\hspace{.5cm}}rl}
       S & \to XSY \cup Z^* & X & \to a\\
       Y & \to b & Z & \to c\\
     \end{array}$
   \end{center}
  (both with start symbol $S$) are isomorphic. They recognize $\{a^n b^n \mid n \in \nat \} \cup L(c^*)$.
\end{example}
\begin{observation}\label{obs:grammar-isomorphism}
  If $G_1$ and $G_2$ are isomorphic, then $L(G_1) = L(G_2)$.
\end{observation}

We want to extend this notion of isomorphism to factorized relations and want to
maintain a property such as Observation~\ref{obs:grammar-isomorphism} which says
that, if the objects are syntactically isomorphic, also their semantics is the
same. 
To this end, for a tuple $t = \langle a_1,\ldots,a_k \rangle$, we denote
the word $a_1 \cdots a_k$ as $\word(t)$. Hence, $\word(\langle \rangle) =
\varepsilon$. For a set $T$ of tuples, we define $\word(T) \defeq \{\word(t)
\mid t \in T\}$.

We now define isomorphisms between uFRs and (star-free) ECFGs. The idea is
analogous as before, but with the difference that the Cartesian product operator
($\times$) in uFRs is replaced by the concatenation operator ($\cdot$) in ECFGs.
That is, we define isomorphisms using a function $h \colon \names \to \names$
that we extend to $X$-expressions as follows:
\[
  \begin{array}{rcll@{\hspace{1cm}}rcl}
    h(\emptyset) &=& \emptyset &&   h(E_1 \times E_2) &=& h(E_1) \cdot h(E_2) \\
    h(\langle  \rangle) &=& \varepsilon &&   h(E_1 \cup E_2) &=& h(E_1) \cup h(E_2) \\
    h(\langle a \rangle) &= & a & \text{ for every $a$ in \values}\\
  \end{array}
\]
A uFR $(N_1,D)$ with $D = \{N_1 \defeq E_1, \dots, N_n \defeq E_n \}$ is \emph{isomorphic} to an ECFG
$G=(T,N,S,R)$ if there is a bijective function $h \colon \{N_1,\dots,N_n\} \to N$ such
that 
\begin{itemize}
\item the start symbol $S$ is $h(N_1)$;
\item for every expression $N_i \defeq E_i$ in $D$, the rule $h(N_i) \rightarrow
  h(E_i)$ is in $R$; and
\item for every rule $N \rightarrow e$ in $R$, there is an expression $N_i
  \defeq E_i$ in $D$ with $h(N_i)=N$ and $h(E_i)=e$.
\end{itemize}

\begin{example}
  Consider the uFR in Example~\ref{ex:FR}, which is visualized in
  Figure~\ref{fig:factorized}. It is routine to check that it is isomorphic to
  the extended context-free grammar in Figure~\ref{fig:CFG}. In fact, the isomorphism
  in this case is the identity function.
\end{example}

We now observe that isomorphisms preserve the size and, in a strong sense, also the
semantics of uFRs and ECFGs. To this end, the \emph{size} $|E|$ of an
$X$-expression or regular expression $E$ is defined to be the number of
occurrences of symbols plus the number of occurrences of operators in $E$. For
example, $|\langle a \rangle| = 1$, $|(\langle a \rangle \cup (\langle \rangle
\times \langle b \rangle)| = 5$, and $|(a \cdot b)^*| = 4$. The size of a uFR
(resp., ECFG) is the sum of the sizes of its $X$-expressions (resp., regular
expressions). 
\begin{proposition}
  If $h$ is an isomorphism from a uFR $(N_1,D)$
  to an ECFG $(T,N,S,R)$ then, for each expression $N_i := E_i$ in $D$ such
  that $A_i = h(N_i)$ and $e_i = h(E_i)$, we have that
  \begin{enumerate}[(a)]
  \item $|e_i| = |E_i|$ and 
  \item $L(A_i) = \word(\sem{N_i})$.
  \end{enumerate}
\end{proposition}

\begin{corollary}
  If a uFR $F$ and ECFG $G$ are isomorphic then they have the same size and 
  $L(G) = \word(\sem{F})$.
\end{corollary}


\subsection{FRs and ECFGs are Isomorphic on Database Relations}\label{sec:uniform-length}

A factorized relation $F$ defines a database relation, where all tuples have the
same arity. An ECFG defining the same relation (i.e., $\word(\sem{F})$) defines
a language in which each word has the same length (which is, in particular, finite).

Let $G = (V,N,S,R)$ be an ECFG. Notice that, if $G$ is trimmed and $L(G)$ is
finite, then $G$ cannot use the Kleene star operator in a meaningful way. Indeed,
it can only use subexpressions $e^*$ if $L(e) = \emptyset$ or $L(e) =
\varepsilon$, in which case $L(e^*) = \varepsilon$. The same holds for
recursion. We therefore assume from now
on that ECFGs that define a finite language are non-recursive and do not use the Kleene star operator.
We call a
nonterminal $A \in N$ \emph{uniform length} if every word $w \in L(A)$ has the
same length. We say that $G$ is \emph{uniform length} if $S$ is uniform length.
We now prove that factorized relations are the same as
uniform-length ECFGs.

\begin{toappendix}
  The proofs for the
  Theorems~\ref{theo:FR-ECFG-uniform},~\ref{theo:FR-ECFG-uniform-disjoint}
  and~\ref{theo:FR-ECFG-variable} are very similar. In particular they all use
  the same mapping $\beta$ from uFRs to ECFGs. Before we do the proofs we like to
  introduce $\beta$ and show that indeed for every uFR $F$ it holds that the ECFG
  $\beta(F)$ is isomorphic to $F$.

  The bijection $\beta$ that we use to map every uFR to an ECFG is induced by the
  identity function on $\names$. Thus, let $(S,D)$ be an uFR with $D=\{N_1
  \defeq E_1, \dots, N_n \defeq E_n\}$. Then
  \begin{align*}
    \beta(S,D) &= (T,N,S,R) \quad \text{where}\\
    T &= \{a \mid \langle  a \rangle \text{ appears in some expression } E_i \text{ in } D\}\\
    N &= \{{ N_1,\dots,N_n}\} \\
    R &= \{ N_1 \rightarrow h(E_1), \dots N_n \rightarrow h(E_n)\}
  \end{align*}

  \begin{lemma}\label{lem:isomorphic}
    Let $F=(S,D)$ be some uFR, then $G=\beta(F)$ is a non-recursive and star-free
    ECFG, such that $F$ and $G$ are isomorphic.
  \end{lemma}
  \begin{proof}
    The statement follows by an easy induction argument.
  \end{proof}

  Sometimes we need the inverse mapping of $\beta$, denoted by $\beta^{-1}$
  which is the following:
  \begin{align*}
    \beta^{-1}(T,N,S,R) &= (S,D) \quad \text{where}\\
    D &= \{ N_1 \defeq h^{-1}(E_1),\dots,N_n \defeq h^{-1}(E_n)\}
  \end{align*}

  Here $h^{-1}$ denotes the inverse mapping of $h$, which is the following:
\[
  \begin{array}{rcll@{\hspace{1cm}}rcl}
    h^{-1}(\emptyset) &=& \emptyset &&   h^{-1}(E_1 \cdot E_2) &=& h^{-1}(E_1) \times h^{-1}(E_2) \\
    h^{-1}(\varepsilon) &=& \langle \rangle &&   h^{-1}(E_1 \cup E_2) &=& h^{-1}(E_1) \cup h^{-1}(E_2) \\
    h^{-1}(a) &= & \langle a \rangle  & \text{ for every $a$ in \values}\\
  \end{array}
\]
\end{toappendix}

\begin{theoremrep}\label{theo:FR-ECFG-uniform}
  There is a bijection $\beta$ between the set of uFRs and the set of
  uniform-length ECFGs
  such that each factorized relation
  $F$ is isomorphic to $\beta(F)$.
\end{theoremrep}
\begin{proof}
  By Lemma~\ref{lem:isomorphic}, it holds that for every uFR $F$, $\beta(F)$ is
  isomorphic to $F$. Furthermore, $\beta$ is clearly a total and
  injective function on the domain of all uFRs and all non-recursive and
  star-free ECFGs. It remains to show that $\beta$ is surjective, i.e., for
  every uniform-length, nonrecursive, and star-free ECFG $G=(T,N,S,R)$, there is
  a uFR $F=(S,D)$ such that $\beta(S,D)=G$.

  Indeed, we let $F= \beta^{-1}(G)$. By definition of $\beta$, it is obvious
  that $\beta(G)=F$. Furthermore $F$ is a valid uFR. As $G$ is non-recursive, the
  names can be topologically sorted in such a way that all name references in
  the expression $E_i$ only reference names $N_j$ with $j>i$. Furthermore, as
  $G$ is uniform length, the arity of each expression is well defined.
\end{proof}

\begin{toappendix}
Since every $N_i$ in $D$ has an associated arity, we have the following
corollary for context-free grammars.
\begin{corollary}
  If $G$ is uniform length, then every non-terminal in $G$ is uniform length.
\end{corollary}
\end{toappendix}

Similarly, we
say that a uniform-length ECFG $G$ \emph{has disjoint positions} if, for every
pair of words $a_1 \cdots a_k \in L(G)$ and $b_1 \cdots b_k \in L(G)$ we have
that $a_i \neq b_j$ if $i \neq j$.
\begin{theoremrep}\label{theo:FR-ECFG-uniform-disjoint}
  There is a bijection $\beta$ between the set of uFRs with disjoint positions and
  the set of uniform-length
  ECFGs with disjoint
  positions such that each factorized relation $F$ is isomorphic to
  $\beta(F)$.
\end{theoremrep}
\begin{proof}
  We use the same mapping $\beta$ as in the proof of
  Theorem~\ref{theo:FR-ECFG-uniform}. Indeed, we can restrict the domain to the
  uFRs with disjoint positions. From the fact that $\beta$ maps each uFR to an
  isomorphic ECFG, we can conclude that the resulting ECFG also has disjoint
  positions. Furthermore, for every uniform length, non-recursive, and star-free
  ECFG $G$ with disjoint positions, the pre-image $\beta^{-1}(G)$ under $\beta$
  is a uFR with disjoint positions. Thus $\beta$ is also bijective when
  restricted to the domain of uFRs with disjoint positions and the image of
  non-recursive and star-free ECFGs with disjoint positions.
\end{proof}

In the remainder of the paper, for an uFR $F$, we will refer to $\beta(F)$ as
\emph{the CFG corresponding to $F$}.


\section{Some Consequences of the Isomorphism}

We now discuss some immediate consequences of
Theorems~\ref{theo:FR-ECFG-uniform} and \ref{theo:FR-ECFG-uniform-disjoint}. We
note that our list is far from exhaustive. In principle, every result on
context-free grammars that holds for uniform-length languages can be lifted to uFRs.
Since uFRs are a class of CFGs due to Theorem~\ref{theo:FR-ECFG-uniform}, we use
some standard terminology for CFGs to uFRs, e.g., the notion of \emph{derivation
trees}. We call a uFR $F$ \emph{deterministic} if the CFG $G$ corresponding to $F$ is
unambiguous, i.e., every word in $L(G)$ has a unique derivation
tree.\footnote{The definition of an nFR $F$ being deterministic 
  in \cite{OlteanuZ-tods15} says that each monomial that can be obtained from
  using distributivity of product over union is distinct. This is equivalent to
  saying that each tuple in $\sem{F}$ has a unique derivation tree.}

\subsection{Membership}
We define the membership problem for uFRs to be the problem that, given a tuple
$t$ and uFR $F$, tests if $t \in \sem{F}$. The CYK
algorithm \cite{HMU-automata} decides membership  context-free languages in
polynomial time.
\begin{corollary}\label{cor:CYK}
  The membership problem for uFRs is in polynomial time.
\end{corollary}

\subsection{FRs versus Non-Deterministic Finite Automata}\label{sec:FRvsNFA}
We call a uFR $(S,D)$ \emph{right-linear} (resp., \emph{left-linear}) if every
expression in $D$ is of the form $A := \langle \rangle$ or $A := \{b\} \times C$
for $A,C \in N$ and $b \in T$ (resp., $A := \langle \rangle$ or $A := C \times
\{b\}$). (The corresponding definition for CFGs is analogous.) Due to
\cite{HMU-automata}, we know that right-linear CFGs are isomorphic to
non-deterministic finite automata. 
The isomorphism, formulated in terms of uFRs, is
that a rule $A \to \{b\} \times C$ corresponds to a transition from
state $A$ to state $C$ with label $b$, and a rule $A := \langle  \rangle$
corresponds to $A$ being an accepting state. A state $A$ is a start state if $A$
does not occur on the right-hand side of a rule in $S$. (The argument for
left-linear uFRs is analogous.)
\begin{corollary}\label{cor:rightlinear}
  Right-linear (and left-linear) uFRs are isomorphic to non-deterministic finite
  automata for uniform-length languages.
\end{corollary}

Let us define the \emph{equivalence problem} of uFRs as follows. Given two uFRs
$F_1$ and $F_2$, is $\sem{F_1} = \sem{F_2}$? Likewise, the \emph{containment
  problem} asks, given two uFRs $F_1$ and $F_2$, whether $\sem{F_1} \subseteq
\sem{F_2}$. Due to \cite[Corollary 5.9]{StearnsH85}, we now know
\begin{corollary}
  Equivalence and containment of deterministic right-linear (resp., left-linear)
  uFRs is in polynomial time.
\end{corollary}
We recall that determinism in uFRs corresponds to \emph{unambiguity} in
context-free grammars and finite automata. In particular, Corollary 5.9
in~\cite{StearnsH85}, which shows that equivalence and containment of
unambiguous finite automata is solvable in polynomial time is non-trivial. In
fact, the result is even more general: it holds for $k$-ambiguous automata for
every constant $k$. Here, $k$-ambiguity intuitively means that each tuple in
$\sem{F}$ is allowed to have up to $k$ derivation trees in $F$.

\subsection{Size Lower Bounds}

Filmus~\cite{Filmus11} proves a 
lower bound on the size of CFGs that, in terms of uFRs, is stated as follows.
\begin{corollary}[\cite{Filmus11}, Theorem 7]\label{cor:lowerbound}
  Consider the set of tuples $S = \{(a_1,\ldots,a_n,b_1,\ldots,b_n) \in
  \{0,1\}^{2n} \mid a_1 \cdots a_n = b_1 \cdots b_n\}$. Then the smallest uFR
  for $S$ has size $\Omega(t^{n/4} / \sqrt{2n})$.
\end{corollary}

In fact, the proofs of Corollary~\ref{cor:lowerbound} and \ref{cor:CYK} use the
fact that CFGs (uFRs) can be brought into \emph{Chomsky Normal Form}
\cite{Chomsky-iandc59}, which may be yet another classical result that is useful
for proving results on uFRs.

\subsection{Counting}

Corollary~\ref{cor:rightlinear} allows us to connect recent
results on counting problems for automata and grammars
\cite{ArenasCJR-jacm21,ArenasCJR-stoc21,MeelCM-pods24} to uFRs. For a class
$\mathcal C$ uFRs, \emph{counting for ${\mathcal C}$} is the problem that,
given a uFR $F \in {\mathcal C}$, asks what is the cardinality of $\sem{F}$.
First, recall that the number of words of a given length $n$ in an unambiguous
CFG (or ECFG) can be counted in polynomial time \cite[Section I.5.4]{FS09}
by a simple dynamic programming approach.
\begin{corollary}
  Counting for deterministic right-linear uFRs is in polynomial time.
\end{corollary}
For general right-linear uFRs, the counting problem is $\#$P-complete, but
the following is immediate from the existence of an FPRAS for $\#$NFA~\cite{ArenasCJR-jacm21} and
Theorem~\ref{theo:FR-ECFG-uniform}.
\begin{corollary}
  Counting for right-linear uFRs admits an FPRAS.
\end{corollary}
Furthermore, the recent more efficient FPRAS for $\#$NFA by Meel et
al.~\cite{MeelCM-pods24} can be applied verbatim to counting for right-linear
uFRs. In fact, Meel et al.~\cite{MeelC-arXiv24} recently generalized the result
to $\#CFG$. The result is still unpublished, but it would imply:
\begin{corollary}
  Counting for uFRs admits an FPRAS.
\end{corollary}
Notice that counting for uFRs (and subclasses thereof) is a practically relevant
question: it is the result of the COUNT DISTINCT query for a factorized
representation.

\subsection{Enumeration}
In terms of enumeration, D\"om\"osi shows that, given a context-free grammar $G$
and length $n$, the set of words in $L(G)$ of length $n$ can be enumerated with
delay polynomial in $n$~\cite{Domosi00}:
\begin{corollary}
  Given an uFR $F$, the set of tuples in $\sem{F}$ can be enumerated with delay
  polynomial in the arity of $F$.
\end{corollary}
We note that, in the case of right-linear uFRs, more efficient algorithms are
possible \cite{AckermanM09,AckermanS09}. The same holds if the uFR is
deterministic~\cite{Piantadosi}. 
  For example, Mu\~{n}oz and
  Riveros \cite{MunozR22} considered \emph{Enumerable Compact Sets (ECS)} as a
  data structure for output-linear delay algorithms. 
  ECS can be viewed as context-free grammars in Chomsky Normal Form for finite languages
  (using a similar isomorphism as in Section~\ref{sec:isomorphism}). In terms of uFRs,
  \cite{MunozR22}  shows that, if the uFR is deterministic and $k$-bounded (which means
  that unions on right-hand sides have constant length), then all its tuples can
  be enumerated in output-linear delay.

\subsection{FRs for Variable-Length Relations}
\label{sec:variable-length}

We now consider a slightly more liberal relational data model in the sense that
a database relation no longer needs to contain tuples \emph{of the same arity}.
The reason why we consider this case is twofold. First, this data model leads to
representations for the \emph{finite languages}, which is a fundamental class in
formal language theory. Second, this data model is the underlying data
model~\cite{reldocs-relations} for the query language Rel~\cite{rel},
implemented by RelationalAI. So it is used in practice. We say that a
\emph{variable-length database relation} is simply a finite set $R$ of tuples.

The correspondence between uFRs and ECFGs in Theorems~\ref{theo:FR-ECFG-uniform} and
\ref{theo:FR-ECFG-uniform-disjoint} allows us to define a natural generalization of FRs for
variable-length relations, which corresponds to ECFGs for finite languages.
  We inductively define \emph{variable-length $X$-expressions} (\emph{vl-$X$-expression}
for short) $E$ as follows:
\begin{itemize}
\item $E=\emptyset$ is a vl-$X$-expression;
\item $E=\langle \rangle$ is a vl-$X$-expression;
\item for each $a \in \V$, we have that $E=\langle a \rangle$ is an $X$-expression with
  $\ar(E)=1$ (singleton);
\item for each $N\in X$, we have that $E=N$ is a vl-$X$-expression (name reference);
\item for vl-$X$-expressions $E_1, \ldots, E_n$ we have that
  \begin{itemize}
  \item $E=(E_1 \cup \cdots \cup E_n)$ is a vl-$X$-expression (union); and
  \item $E=(E_1 \times \cdots \times E_n)$ is a vl-$X$-expression (Cartesian product).
  \end{itemize}
\end{itemize}
\begin{definition}
A \emph{variable-length factorized relation (vlFR)} is a pair $(S,D)$, where $S \in \names$ is
the start symbol and $D = \{N_1 \defeq E_1, \ldots , N_n \defeq E_n\}$ is a set
of expressions where:
\begin{enumerate}
\item $N_1 = S$ and $\ar(N_1 = k)$;
\item Each $N_i$ is an expression name;
\item Each $E_i$ is an $X_i$-expression for $X_i=\{N_{i+1},\dots,N_{n}\}$; and
\item $\ar(N_{i})=\ar(E_i)$ for all $i=1,\dots,n$.
\end{enumerate}
\end{definition}
The semantics $\sem{F}$ of a variable-length factorized relation $F$ is defined analogously as
for FRs. The difference is that the result is now a variable-length database
relation. 

\begin{theoremrep}\label{theo:FR-ECFG-variable}
  There is a bijection $\beta$ between the set of vlFRs and
  the set ECFGs for finite languages such that each factorized representation $F$ is isomorphic to
  $\beta(F)$. 
\end{theoremrep}
\begin{proof}
  We again use the same mapping $\beta$ as in the proofs of
  Theorems~\ref{theo:FR-ECFG-uniform} and~\ref{theo:FR-ECFG-uniform-disjoint}. The only
  difference is that we extend the domain of $\beta$ to FRs of variable-length
  and the image of $\beta$ to star-free non-recursive ECFGs for finite
  languages. Lemma~\ref{lem:isomorphic} clearly extends to these classes. Thus
  for each FR $F$, $\beta(F)$ is isomorphic fo $F$, even if $F$ is not
  uniform-length. The argumentation that $\beta$ is a bijection is analogous to
  the other cases.
\end{proof}


\section{Path Representations in Graph Databases}\label{sec:pmrs}

\nc{\nid}{\mathsf{NID}}
\nc{\eid}{\mathsf{EID}}
\nc{\bL}{\textbf{L}}
\nc{\pathpred}{\textsf{Path}}
\nc{\paths}{\ensuremath{\mathsf{Paths}}}
\nc{\spaths}{\mathsf{SPaths}}
\nc{\mpaths}{\mathsf{MPaths}}
\nc{\set}{\mathsf{Set}}
\nc{\src}{\mathsf{src}}
\nc{\tgt}{\mathsf{tgt}}
\nc{\pmr}{\text{PMR}\xspace}
\nc{\pmrs}{\text{PMRs}\xspace}

\nc{\multileft}{\ensuremath{\{\!\!\{}}
\nc{\multiright}{\ensuremath{\}\!\!\}}}

We now start exploring the relationship between uFRs and \emph{Path
  Multiset Representations (PMRs)}, which were recently introduced as a succinct
data structure for (multi)sets of paths in graph
databases~\cite{MartensNPRVV-vldb23}.
Several studies
demonstrate that they can drastically speed up evaluation of queries that
involve regular path queries with path variables
\cite{MartensNPRVV-vldb23,pathfinder,pathfinderdemo}.

We briefly introduce PMRs and explain their connection to finite automata. This
allows us to relate them to uFRs and study some size tradeoffs later in this section.

\subsection{Path Multiset Representations}
We use
edge-labeled multigraphs as our abstraction of a graph database.\footnote{PMRs
  were originally defined on \emph{property graphs}, which are more complex
  than edge-labeled graphs. The definition of PMRs presented here is therefore
  slightly simplified.} A \emph{graph database} is a tuple $G = (N,E,\lab)$,
where $N$ is a finite set of \emph{nodes}, $E \subseteq N \times N$ is a finite
set of \emph{edges}, and $\lab \colon E \to \values$
is a function that associates a
label to each edge. A \emph{path} in $G$ is a sequence of nodes
$u_0,\ldots,u_n$, where $(u_{i-1},u_{i}) \in E$ for every $i \in [n]$.
\begin{definition}
  A \emph{path multiset representation (PMR)} over a graph database $G =
  (N_G,E_G,\lab_G)$ is a tuple $R = (N,E,\gamma,S,T)$ where
  \begin{itemize}
  \item $N$ is a finite set of nodes;
  \item $E \subseteq N \times N$ is a finite set of edges;
  \item $\gamma : N \to N_G$ is a homomorphism (that is, if $(u,v) \in E$,
    then $(\gamma(u),\gamma(v)) \in E_G$);
  \item $S \subseteq N$ is a finite set of \emph{start nodes};
  \item $T \subseteq E$ is a finite set of \emph{target nodes}.
  \end{itemize}
\end{definition}

The semantics of PMRs is defined as follows. They can be used as a
representation of a \emph{set} or a \emph{multiset} of paths. More precisely, we
define $\spaths(R) = \{\gamma(p) \mid p$ is a path from some node in $S$ to some
node in $T$ in $R\}$. Notice that each $\gamma(p)$ is indeed a path in $G$,
since $\gamma$ is a homomorphism. $\mpaths(R)$ is defined similarly, but it is a multiset,
where the multiplicity of each path $p$ in $G$ is the number of paths $p'$ in
$R$ such that $\gamma(p') = p$.

\begin{figure}
  \begin{subcaptionblock}{.3\textwidth}
    \centering
    \begin{tikzpicture}[state/.style={circle,draw,inner sep=1pt}]
      \node[state] (A) {$A$};
      \node[state] (B) at ($(A)+(1.3,.7)$) {$B$};
      \node[state] (C) at ($(A)+(1.3,-.7)$) {$C$};
      \node[state] (D) at ($(B)+(1.3,0)$) {$D$};
      \node[state] (E) at ($(D)+(0,-1.4)$) {$E$};

      \path[-latex]
      (A) edge[bend left] node[above left]{$a$} (B)
      (B) edge[bend left] node[right]{$a$} (C)
      (C) edge[bend left] node[below left]{$a$} (A)
      (B) edge node[above]{$a$} (D)
      (D) edge node[right]{$a$} (E)
      (C) edge node[above]{$a$} (E)
      ;
    \end{tikzpicture}
    \caption{An edge-labeled graph $G$\label{fig:pmr-a}}
  \end{subcaptionblock}
  \qquad
  \begin{subcaptionblock}{.5\textwidth}
    \begin{tikzpicture}[state/.style={rectangle,rounded corners,draw}]
      \node[state] (B2) {$B_2$};
      \node[state] (C2) at ($(B2)+(1.3,.7)$)  {$C_2$};
      \node[state,fill=yellow] (A1) at ($(C2)+(1.3,0)$)   {$A_1$};
      \node[state] (B1) at ($(A1)+(1.3,-.7)$) {$B_1$};
      \node[state] (C1) at ($(A1)+(0,-1.4)$)  {$C_1$};
      \node[state] (A2) at ($(C2)+(0,-1.4)$)  {$A_2$};
      \node[state,fill=Green] (D) at ($(B1)+(1.8,0)$)      {$D$};
      \node[state] (B3) at ($(B1)+(0.7,.7)$)  {$B_3$};

      \path[-latex]
      (A1) edge[bend left] (B1)
      (B1) edge[bend left] (C1)
      (C1) edge (A2)
      (A2) edge[bend left] (B2)
      (B2) edge[bend left] (C2)
      (C2) edge  (A1)
      (A1) edge[bend left]  (B3)
      (B1) edge  (D)
      (B3) edge  (D)
      ;
    \end{tikzpicture}
    \caption{A PMR $R$ of the paths of even length from $A$ to $D$ (with multiplicity
      two for the shortest path)\label{fig:pmr-b}}
  \end{subcaptionblock}
  \caption{An edge-labeled graph $G$ and a PMR for a multiset of its paths\label{fig:pmr}}
\end{figure}

\begin{example}
  Consider the graph in Figure~\ref{fig:pmr-a}. Figure~\ref{fig:pmr-b} depicts a
  PMR $R$ representing the set of paths from $A$ to $D$ in $G$ that have even
  length. The homomorphism $\gamma$ simply matches both nodes $A_1$ and $A_2$ to
  $A$; and similarly for the other indexed nodes. We define $S = \{A_1\}$ and $T
  = \{D\}$.  Following our definition, we
  now have that $\spaths(R)$ is indeed the set of paths from $A$ to $D$ in $G$
  that have even length. In $\mpaths(R)$, the shortest such path (having length
  two) has multiplicity two, whereas all other paths have multiplicity one.
\end{example}

\subsection{PMRs versus Finite Automata}

PMRs are closely connected to finite automata by design.
One reason for this design choice is that graph pattern matching in languages such as Cypher, SQL/PGQ,
and GQL starts with the evaluation of \emph{regular path queries}, which match
``regular'' sets of paths.

We explain this connection next and assume familiarity with finite automata. In the following, we will denote
\emph{nondeterministic finite automata (NFAs)} as tuples $A = (\Sigma, Q, \delta, I,
F)$ where $\Sigma$ is its \emph{symbols} (or \emph{alphabet}), $Q$ its finite set of \emph{states},
$\delta$ its set of \emph{transitions} of the form $q_1 \xrightarrow{a} q_2$
(meaning that, in state $q_1$, the automaton can go to state $q_2$ by reading
the symbol $a$), $I \subseteq Q$ its set of \emph{initial states}, and $F
\subseteq Q$ its set of \emph{accepting states}. As usual, we denote by $L(A)$
the \emph{language} of $A$, which is the set of words accepted by $A$.

The connection between PMRs and NFAs is very close.
Indeed, we can turn a PMR $R = (N,E,\gamma,S,T)$ over graph $G = (N_G,E_G,\lab)$
into an NFA $N_R = (\Sigma, Q, \delta, I, F)$ where
\begin{enumerate}
\item the alphabet $\Sigma$ is the set $N_G$ of nodes of $G$;
\item the set $Q$ of states is $N \cup \{s\}$, where we assume $s \notin N$;
\item for every edge $e = (u,v) \in E$, there is a transition
  $u \xrightarrow{\gamma(v)} v$;
\item for every node $u \in S$, we have a transition $s
  \xrightarrow{\gamma(u)} u$;
\item $I = \{s\}$; and $F = T$.
\end{enumerate}
As usual, we denote by $L(N_R)$ the set of words accepted by $N_R$. The
automaton $N_R$ accepts precisely the set of paths represented by $R$.
\begin{propositionrep}[Implicit in \cite{MartensNPRVV-vldb23}]\label{prop:PMR-vs-automaton-set}
  $L(N_R) = \spaths(R)$.
\end{propositionrep}
\begin{proof}
  This is a direct corollary of
  Proposition~\ref{prop:PMR-vs-automaton-multiset}. Therefore we refer to that proof.
\end{proof}
In fact, there also exists a \emph{multiset} language of NFAs, denoted $ML(A)$, in which the
multiplicity of each word $w \in ML(A)$ is the number of accepting runs that $A$
has on $w$. Analogously to Proposition~\ref{prop:PMR-vs-automaton-set}, one can
show the following.
\begin{propositionrep}[Implicit in \cite{MartensNPRVV-vldb23}]\label{prop:PMR-vs-automaton-multiset}
  $\mathit{ML}(N_R) = \mpaths(R)$.
\end{propositionrep}
\begin{proof}
  We start with the $\subseteq$ direction. Assume that $v_1 \cdots v_n$ is
  accepted by $N_R$. By construction there is an accepting run
  \[
    s \xrightarrow{\gamma(u_1)} u_1 \xrightarrow{\gamma(u_2)} u_2 \xrightarrow{\gamma(u_3)} \cdots
    \xrightarrow{\gamma(u_n)} u_n\;,
  \]
  such that $\gamma(u_i)=v_i$ for $1 \leq i \leq n$. Therefore, the edges
  $(u_1,u_2),\dots,(u_{n_1},u_n)$ are in $E$. Thus $\gamma(u_1) \cdots
  \gamma(u_n)=v_1 \cdots v_n$ is in $\mpaths(R)$. Furthermore, distinct
  accepting runs for $v_1 \cdots v_n$ in $N_R$, translate to distinct paths in
  $R$. Thus the multiplicity of $v_1 \cdots n_v$ in $\mpaths(R)$ is at least as
  big as the number of accepting runs for $v_1 \cdots v_n$ in $N_R$.

  We continue with the $\supseteq$ direction, which works analogously. Assume
  that $v_1 \cdots v_n$ is a path in $\mpaths(R)$. Then there are edges
  $(u_1,u_2), \dots (u_{n_1},u_n)$ in $E$ such that $\gamma(u_i)=v_i$ for $1
  \leq i \leq n$.. By construction, there is an accepting run 
  \[
    s \xrightarrow{\gamma(u_1)} u_1 \xrightarrow{\gamma(u_2)} u_2 \xrightarrow{\gamma(u_3)} \cdots
    \xrightarrow{\gamma(u_n)} u_n
  \]
  in $N_R$. Therefore the word $\gamma(u_1) \cdots \gamma(u_n) = v_1 \cdots v_n$
  is accepted by $N_R$. Furthermore, distinct paths in $R$ that are mapped to
  the same path $v_1 \cdots v_n$ in $G$ translate to distinct runs in $N_R$.
  Thus the multiplicity of $v_1 \cdots v_n$ in $\mathit{ML}(N_R)$ is at least as big as
  in $\mpaths(R)$.

  Taking both directions together, we have shown $\mathit{ML}(N_R) = \mpaths(R)$.
\end{proof}

The correspondence in Proposition~\ref{prop:PMR-vs-automaton-multiset} is
interesting for the purposes of representing path multisets, because deciding
for given NFAs $A_1$ and $A_2$ if $\mathit{ML}(A_1) = \mathit{ML}(A_2)$ is in polynomial
time~\cite{Tzeng-ipl96} if the NFAs do not have $\varepsilon$-transitions,
which is the case here. Deciding if $L(A_1) = L(A_2)$ on the other hand is
\pspace-complete \cite{MeyerS-focs72}.
(The same complexities hold for the respective containment problems.)
This can be interesting if we want to consider questions like finding
optimal-size representations of a (multi)set of paths.

\subsection{Comparing uFRs and PMRs}
In this section, we identify a path $u_1,\ldots,u_n$ with the database tuple
$(u_1,\ldots,u_n)$. Furthermore, for an uFR $F$ and PMR $R$, we compare
$\sem{F}$ with $\spaths(R)$, i.e., we compare them under set semantics. (A
similar comparison can be made when considering them both under bag semantics.)
Under these assumptions, we have the following observation.
\begin{proposition}
  PMRs are strictly more expressive than uFRs.
\end{proposition}
\begin{proof} Every uFR represents a finite relation, which can be represented
  by a PMR (in formal language terms: every finite language is regular).
  Furthermore, PMRs can represent some infinite relations, namely those whose
  corresponding word language is regular~\cite{MartensNPRVV-vldb23}.
\end{proof}

It is interesting to compare the relative size of PMRs and uFRs. Indeed, most practical query languages (e.g., GQL, Cypher)
use keywords to ensure that the sets of paths to be considered in graph
pattern matching are finite (SHORTEST, TRAIL, SIMPLE, ACYCLIC).\footnote{It is an
interesting question if such keywords that force finite number of paths are
indeed always needed, and PMRs show one way to finitely represent an infinite number of
paths.}
This means that, when using these languages, one can in principle use PMRs as
well as uFRs to represent sets of paths.


Using uFRs to represent sets of paths in graph database systems opens up a wide
array of questions. More precisely, context-free grammars (CFG), unambiguous context-free
grammars (UCFG), non-deterministic finite automata (NFA), unambiguous finite
automata (UFA), and deterministic finite automata (DFA) for finite-length
languages are all special cases of
uFRs and are all able to represent some set of $n$ tuples using a representation of
size $O(\log n)$. All these formalism have different
properties. (E.g., counting for UFA is easy by~\cite[Corollary 5.9]{StearnsH85}, but
$\#$P-complete for NFA, we know that we can convert an NFA into a DFA in
exponential time, etc.) So, which representation can be used in which case?
This question actually calls for an investigation that is too extensive for one
paper --- here we investigate the size tradeoffs between the different models.

\begin{figure}
  \centering
  \begin{tikzpicture}[node distance = 2.5cm]
    \node (CFG) {CFG};
    \node [right of=CFG]  (UCFG) {UCFG};
    \node [below of=CFG] (NFA) {NFA};
    \node [right of=NFA]  (UFA) {UFA};
    \node [right of=UFA]  (DFA) {DFA};
    \node [right of=DFA]  (Trie) {Set};

    \path[->,very thick,red]
    (CFG) edge (UFA)
    (UCFG) edge [bend left] (DFA)
    ;

    \path[->,dashed,very thick,red]
    (CFG) edge (UCFG)
    ;

    \path[->,very thick,blue]
    (CFG) edge (NFA)
    (NFA) edge (UFA)
    (NFA) edge[dashed,out=60,in=210] (UCFG)
    (UCFG) edge[out=240,in=30] (NFA)
    (UCFG) edge (UFA)
    (UFA) edge (DFA)
    (DFA) edge (Trie)
    ;
  \end{tikzpicture}
  \caption{Worst-case unavoidable blow-ups for succinct representations of
    uniform length relations. Every path that consists of only blue edges
    represents an unavoidable exponential blow-up and every path that contains
    at least one red (solid) edge represents an unavoidable
    double exponential blow-up. If there is no path, then there exists a linear
    translation. For the dashed edges, we only prove an upper bound. The
    corresponding lower bounds are conditional on Conjecture~\ref{conjecture:UCFG}.}
\end{figure}

\paragraph*{Uniform-Length Relations}

We say that there is an \emph{$f(n)$-size translation} from one model $X$ to a
model $Y$, if there exists a translation from $X$ to $Y$ such that each object
of size $n$ in $X$ can be translated to an equivalent object of size $O(f(n))$
in $Y$. We say that these translations are \emph{tight} if there is an
infinite family of objects in $X$ in which, for each object of size $n$, the
smallest equivalent object in $Y$ has size $\Omega(f(n))$. When $F$ is a set of
functions (such as the exponential or double-exponential functions), we say that
there's an $F$-translation if there exists a function $f \in F$ such that there
is an $f(n)$-size translation. Again, we say that the translations are tight, if
there is a function $f \in F$ for which the translation is tight.

\begin{toappendix}
  To keep the proofs somewhat compact, we first introduce the standard
  translations between language models and only keep the size bounds for the
  proofs.

  We start with the translation from non-recursive and star-free ECFGs to
  star-free regular expressions. Let $G=(T,N,S,R)$ be some non-recursive and
  star-free ECFG. We first, convert $G$ into an ECFG $G'$ that only uses a
  single non-terminal $S$ as follows: As long, as the expression of $S$ contains
  some nonterminal $N$, we replace $N$ by its expression. As the grammar is
  non-recursive, this process terminates. The resulting regular expression $E$
  of $S$ satisfies $L(E)=L(G)$.

  We continue with the Glushkov construction of NFAs from regular expressions.
  Let $R$ be some regular expression. We let $Q$ be the set of all alphabetic
  symbols in $R$ and one additional state $q_0$ that is the initial state. We
  add a transition $q_0 \xrightarrow{a} q_a$ from $q_0$ to some occurrence $q_a
  \in Q$ of the $a$-symbol, if this symbol can be matched at the beginning of a
  word. Similarly we add a transition $q_a \xrightarrow{b} q_b$ from some
  occurrence $q_a$ of an $a$-symbol to some occurrence $q_b$ of a $b$-symbol, if
  the specific occurrence of the $b$-symbol can be matched after the specific
  occurrence of the $a$-symbol in some word. The accepting states are all
  occurrences of symbols that can be matched to the end of some word.

  We already introduced UCFGs and UFAs. We now introduce unambiguous (star-free)
  regular expressions. A star-free regular expressions $R$ is unambiguous, if one
  of the following is true:
  \begin{itemize}
  \item $R$ is an atomic regular expression;
  \item $R$ is $R_1 \cup R_2$, where $R_1$ and $R_2$ are unambiguous regular
    expressions that satisfy $L(R_1) \cap L(R_2) = \emptyset$; or
  \item $R$ is $R_1 \cdot R_2$, where $R_1$ and $R_2$ are unambiguous regular
    expressions such that there are no words $u_1u_2=v_1v_2$ with $u_1\neq v_1$
    and $u_1,v_1 \in L(R_1)$ and $u_2, v_2 \in L(R_2)$.
  \end{itemize}

  We note that an ECFG is unambiguous if and only if its associated regular
  expression (by the construction above) is unambiguous.
  \end{toappendix}

\begin{theoremrep}\label{theo:size-uniform}
  Over uniform-length languages, there are exponential translations
  \begin{enumerate}[(E1)]
  \item from CFG to NFA;
  \item from UCFG to UFA;
  \item from NFA to UCFG;
  \item from NFA to UFA;
  \item from UFA to DFA;
  \item from DFA to Set; and
  \item from NFA to Set.
  \end{enumerate}
 There are double-exponential translations
  \begin{enumerate}[(DE1)]
  \item from CFG to UFA;
  \item from CFG to UCFG;
  \item from UCFG to DFA; and
  \item from CFG to Set.
  \end{enumerate}
  Furthermore, these translations are tight for (E1--E2,E4--E7) and (DE1,DE3--DE4).
\end{theoremrep}
We conjecture that the translations for (DE2) and (E3) are also tight. In fact,
they are tight under Conjecture~\ref{conjecture:UCFG}. To the best of our
knowledge, the literature does not yet have well-developed methods for proving
size lower bounds for UCFGs. A proof of Conjecture~\ref{conjecture:UCFG} using
communication complexity has
recently been claimed by Mengel and Vinall-Smeeth~\cite{MengelVS25}.
\begin{proof}
  We start with the upper bounds. There cannot be a blow-up that is more than
  double-exponential even starting from the most succinct representation by a
  CFG. If a CFG accept any word that is longer than exponential in the number of
  non-terminals of the grammar, then by the pumping lemma of context free
  languages, the grammar also accepts words of arbitrary length, contradicting
  the assumption that all words have uniform length. Furthermore, there are not
  more than doubly exponential many words of exponential length. Therefore,
  there can never be a more than double-exponential blowup and we conclude (DE1)
  to (DE3).

  We now describe single exponential translations for (E1) to (E6). We start
  with (E1). The translation described above from ECFGs to regular expressions
  has an at most exponential blowup, as the nesting depth of $G$ is bounded by
  the number of non-terminals. We can convert $E$ into some NFA $A$ with
  linearly many states in $|E|$ using the well known Glushkov construction.

  For (E2), we can use the very same construction. Indeed, if the grammar $G$ is
  unambiguous, then the resulting expression $E$ is also unambiguous. It is well
  known that the Glushkov automaton for an unambiguous regular expression is
  unambiguous.

  We continue with (E7). Indeed an NFA $A$ for a finite language can only accept
  strings up to length $|A|$ according to the pumping lemma. As there are only
  most exponentially many different words of linear length, the set has at most
  exponential size.

  For (E3--6) we observe that we can easily convert any set to a linear size
  DFA or UCFG. The upper bound follows from (E7).

  We now continue to show the lower bounds. We start with (DE1). Consider the
  family of languages 
  \[
    L^1_n = \{(a+b)^{k}a(a+b)^{2^n}a(a+b)^{2^n-k} \mid k,n \in \nat, k \leq 2^n\}
  \]
  With other words, $L^1_n$ is the language of all words of length $2^{2n}+2$ that
  have two $a$ positions with distance exactly $2^n+1$.

  Intuitively, a CFG for $L^1_n$ guesses the $a$-positions and then verifies that
  at these positions there are indeed $a$-symbols. Formally, we use the
  following CFG $G^1_n$ that accepts the language $L^1_n$:

  \begin{align*}
    A_i &\;\;\rightarrow\;\; B_{i-1} A_{i-1} + A_{i-1}B_{i-1} & \text{for all }0 < i \leq n \\
    A_0 &\;\;\rightarrow\;\;  B_0aB_na + aB_naB_0  \\
    B_i &\;\;\rightarrow\;\; B_{i-1}B_{i-1} & \text{for all } 0 < i \leq n  \\
    B_0&\;\;\rightarrow\;\; a + b 
  \end{align*}
  We use $A_n$ as starting symbol. The rule for $A_0$ constructs the middle part
  of the word ensuring there are two positions with $a$-symbols with exactly
  $2^n$ positions in between. We note that $B_n$ produces all strings of length $2^n$.
  The rules for $A_0$ to $A_k$ add $2^n$ additional positions to the left and right.

  In~\cite[Theorem~5]{Leung2005}, it is shown that every UFA for the language $L^1_n$ needs at
  least doubly exponentially many states proving (DE1).
  
  We conjecture (see Conjecture~\ref{conjecture:UCFG}) that the smallest UCFG
  for $L^1_n$ is also of double exponential size in $n$. If this conjecture holds
  true, the double exponential blow-up (DE2) is tight. Furthermore this would
  also yield an exponential lower bound for the translation (E3), as there is an
  NFA of exponential size in $n$ for $L^1_n$, i.e., exponentially smaller than the
  conjectured lower bound for UCFGs.
  
  We continue with (DE3). We consider the family of languages \[L^2_n = (a+b)^i a
  (a+b)^{2^n}c(a+b)^{2^n - i-1}\;.\] This is the language of all words of
  length $2^{2n}+1$ which have exactly one $c$ and $2^n+1$ positions left of the
  $c$ there has to be an $a$. Any DFA for $L^2_n$ needs doubly exponentially many
  states in $n$, as it needs to remember the last $2^n$ positions. We now provide
  an unambiguous grammar $G^2_n$ for $L^2_n$ of linear size in $n$ proving the double exponential blowup.
  \begin{align*}
    A_i &\;\;\rightarrow\;\; B_{i-1} A_{i-1} + A_{i-1}B_{i-1} & \text{for all }0 < i \leq n \\
    A_0 &\;\;\rightarrow\;\;  aB_nc \\
    B_i &\;\;\rightarrow\;\; B_{i-1}B_{i-1} & \text{for all } 0 < i \leq n  \\
    B_0&\;\;\rightarrow\;\; a + b 
  \end{align*}
  We use $A_n$ as starting symbol. The rule for $A_0$ constructs the middle part
  of the word ensuring there is exactly one $c$-position and that there is an
  $a$-position in the correct distance to the left. The rules for $A_1$ to $A_n$
  add $2^n-1$ additional positions to the left and right.

  We now argue that $G^2_n$ is unambiguous, i.e., that for any string in $L^2_n$
  there is exactly one syntax tree in $G^2_n$. We note that only the rules for
  $A_i$ with $i>0$ and the rule for $B_0$ have a choice. For a given input
  string $w$, the unique $c$-position in $w$ determines which of the two choices
  has to be taken for every $A_i$-rule in much the same way as a given number
  identifies the values for all positions in its binary representation. Finally,
  the choice in each application of the rule for $B_0$ is uniquely determined by
  the alphabet symbol at the corresponding position in the string.

  The lower bound for (DE4) easily follows from (DE1). Alternatively one could
  use the language $(a+b)^{2^n}$ for which there is a CFG of linear
  size in $n$.

  We continue with the lower bound for (E1) and (E2). We consider the family of
  languages $L^3_n=\{ww^{\mathsf{REV}} \mid w \in (a+b)^n\}$, where
  $w^{\mathsf{REV}}$ denotes the reversal of $w$. The language $L^3_n$ is easily
  derived by the following grammar $G^3_n$:

  \begin{align*}
    A_i &\;\;\rightarrow\;\; aA_{i-1}a + bA_{i-1}b & \text{for all } 1 < i \leq n  \\
    A_1&\;\;\rightarrow\;\; aa + bb
  \end{align*}

  The grammar $G^3_n$ is an (unambiguous) context free grammar for the language
  $L^3_n$. As any automaton for $L^3_n$ needs at least $2^n$ states to remember the
  word $w$, we get a lower bound for the translation from (U)CFG to NFA and
  therefore also to UFA. This shows the lower bounds for (E1) and (E2).

  For (E4) we use the class of languages $L^4_n = \{(a+b)^k a (a+b)^n a
  (a+b)^{n-k} \mid k,n \in \nat, k \leq n \}$. One can construct a
  linear size NFA for $L^4_n$ which guesses the $a$-position that has another
  $a$-position in the correct distance. In~\cite[Theorem~5]{Leung2005} it is shown that every
  UFA for $L^4_n$ needs to have exponentially many states in $n$.
  
  For (E5) we use the class of languages $L^5_n=\{(a+b)^{i}a(a+b)^nc(a+b)^{n-i}
  \mid 0 \leq i \leq n\}$. Any DFA for $L^5_n$ needs at least $2^n$ states to
  remember up to $n$ symbols. However a UFA can guess at which position the $c$
  occurs and then verify both, the guess and that $a$ occurs in the correct distance.
  This is a minor variation of the classic text book example showing that UFAs
  can be exponentially more succinct than DFAs adapted to uniform-length
  languages.

  The lower bounds for (E6) and (E7) clearly follow using the class of languages
  $L^6_n=(a+b)^n$ that contain all words of length $n$.
\end{proof}

\begin{conjecture}\label{conjecture:UCFG}
  For each $n \in \nat$, the smallest unambiguous
  context-free grammar for the language \[L_n = \{(a+b)^{k}a(a+b)^na(a+b)^{n-k} \mid
    k \leq n\}\] of words of length $2n+2$ has size
  $2^{\Omega(n)}$.
\end{conjecture}

One reason why we believe in Conjecture~\ref{conjecture:UCFG} is
because there does not exist a UCFG for the generalized version of the language
to the unbounded length setting.
\begin{propositionrep}
   There does not exist a UCFG for the (infinite) language 
   \[L = \{(a+b)^n a (a+b)^n a (a+b)^{n-k} \mid k,n \in \nat, k \leq n\}\;.\]
 \end{propositionrep}
\begin{proof}
  We use Ogden's Lemma~\cite{Ogden1968} to show the inherent ambiguity of the language.
  We assume that there is an unambiguous CFG $G$ for $L$. 
  Let $n$ be the constant from Ogden's Lemma for $G$. We choose the word
  $b^{n}ab^{n!+3n+2}ab^{n}ab^{n!+n}b$ and mark the first block of $b$-symbols.
  We note that only the first and second $a$-position are in the correct distance.

  According to Ogden's Lemma, there has to be be a derivation
  \[
    S \rightarrow^* uAy \rightarrow^* uvAxy \rightarrow^*  uvwxy\;.
  \]
  We assume that pumping happens only in $b$-blocks, i.e., $a$-symbols cannot be
  pumped. This can be enforced by intersecting $L$ with the regular language of
  all words with exactly three $a$-symbols. Unambiguous languages are closed
  under intersection with regular languages.

  By choosing to mark the first block of $b$-symbols, we ensure that pumping
  actually happens in this block. And by the definition of the language, pumping
  also has to happen in the second block of $b$-symbols as otherwise we can
  reach a word not in $L$ by choosing $uv^0wx^0z$. It is impossible that this
  word is short enough for the second and third $a$-position to be in the
  correct distance and the distance between the first and second $a$-position
  would not change.

  Thus $v$ has to contain at least one and at most $n$ $b$-symbols from the
  first block and $x$ has to contain the same number of $b$-symbols from the
  second block. We can derive the word $b^{n+n!}ab^{2n!+3n+2}ab^{n}ab^{n!+n}b$
  by choosing $uv^{n!/|v|+1}wx^{n!/|v|+1}y$.

  Now we use the word $b^{n!+n}ab^{n}ab^{n!+3n+2}ab^{n}b$ and mark the last
  block of $b$-symbols. Analogous reasoning yields that we can pump the third
  and forth block of $b$-symbols and derive the same word
  $b^{n+n!}ab^{2n!+3n+2}ab^{n}ab^{n!+n}b$. The two derivations are clearly
  distinct. Thus $G$ is not unambiguous contradicting our assumption.
\end{proof}

\begin{toappendix}
  \newcommand{\todisjoint}{{\mathsf{mark}}\xspace}
  In order to denote languages with disjoint positions we use the function
  $\todisjoint$ which, given a uniform length language $L$ over some alphabet
  $\Sigma$ produces a uniform length language with disjoint positions over
  alphabet $\Sigma \times \{0,\dots,n-1\}$, where $n$ is the length of every word
  in $L$. The language $\todisjoint(L)$ results from $L$ by annotating each
  symbol in each word with its position in the word starting from zero.

  We use the following Lemma that bounds the blow-up for translating a grammar
  $G$ for some uniform length language $L$ into a grammar for $\todisjoint(L)$.
  The intuitive idea is to encode the starting position of each subword in each
  non-terminal.
  \begin{lemma}\label{lem:disjoint-grammars}
    Let $L$ be a language such that all words have length $n$, and let $G$ be a
    CFG for $L$. Then there is an CFG $G'$ for $\todisjoint(L)$ of size at most
    $n \cdot |G|$. Furthermore, $G'$ is unambiguous if $G$ is unambiguous.
  \end{lemma}
  \begin{proof}
   Let $G=(T,N,S,R)$ be some CFG for $L$. As all words have the same length
    $n$, for each non-terminal $A$ it holds that all words produced by $A$ have
    the same length. We denote this length by $\len(A)$. Now we use $T \times
    \{0,\dots,n-1\}$ as the new set of terminal symbols and $N \times
    {0,\dots,n-1}$ as the new set of non-terminals.
    
    The second entry $i$ of some non-terminal $(A,i)$ shall encode the length of
    the derived word that is left to the subword derived by $(A,i)$. The new
    start symbol is $(S,0)$. For each rule $A \rightarrow \alpha_1 \cdots
    \alpha_k$, we introduce new rules 
    \[(A,i) \rightarrow (\alpha_1,i) (\alpha_2,i+\len(\alpha_1))
      (\alpha_3,i+\len(\alpha_1\alpha_2)) \cdots (\alpha_k, i +
      \len(\alpha_1 \cdots \alpha_{n-1}))\]
    for each $0 \leq i \leq n-(\len(\alpha_1\cdots\alpha_n))$.

    It is straightforward to verify that the resulting grammar indeed accepts
    $\todisjoint(L)$.
  \end{proof}

  As automata for uniform length languages always satisfy that the position is
  implied by the state from which a transition originates, we get the following
  observation:
  \begin{observation}\label{obs:disjoint-automata}
    There is no blow-up for translating finite automata for a language $L$ with
    uniform length to the language $\todisjoint(L)$ (or vice versa).
  \end{observation}
  
\end{toappendix}

\begin{theoremrep}\label{theo:size-uniform-disjoint}
  Over uniform-length languages with disjoint positions, there are exponential
  translations
  
  \noindent
  \begin{tabular}{l@{\hspace{1.5cm}}l}
    \emph{\textbf{(E1)}} from CFG to NFA; & \emph{\textbf{(E5)}} from NFA to UFA;\\
    \emph{\textbf{(E2)}} from UCFG to UFA; & \emph{\textbf{(E6)}} from UFA to DFA;\\
    \emph{\textbf{(E3)}} from CFG to UCFG; & \emph{\textbf{(E7)}} from DFA to Set; and\\
    \emph{\textbf{(E4)}} from NFA to UCFG; & \emph{\textbf{(E8)}} from CFG to Set.
  \end{tabular}
  
  \noindent Furthermore, the translations (E1--E2,E5--E8) are tight.
\end{theoremrep}
 Again, if Conjecture~\ref{conjecture:UCFG} holds true, the translations (E3--4)
 are also tight. 
\begin{proof}
  The upper bounds all follow from
  Observation~\ref{obs:disjoint-positions-blowup}, i.e., from the fact that a
  set of tuples is of at most exponential size in the arity.

  For the lower bounds, we only list the used classes of languages that are defined
  based on languages from the proof of Theorem~\ref{theo:size-uniform}. The
  bounds then follow from the argumentation in the proof of the theorem together
  with Lemma~\ref{lem:disjoint-grammars} and
  Observation~\ref{obs:disjoint-automata}.
  \begin{itemize}
  \item $\todisjoint(L^3_n)$ for (E1) and (E2);
  \item $\todisjoint(L^4_n)$ for (E5);
  \item $\todisjoint(L^5_n)$ for (E6); and
  \item $\todisjoint(L^6_n)$ for (E7) and (E8).
  \end{itemize}

  We exemplify the usage of Lemma~\ref{lem:disjoint-grammars} and
  Observation~\ref{obs:disjoint-automata} on (E1). By (the proof of)
  Theorem~\ref{theo:size-uniform} we know that the smallest NFA for $L^3_n$ is
  exponentially larger than the smallest CFG. All words in the language have
  length linear in $n$. By Lemma~\ref{lem:disjoint-grammars}, we conclude that
  the smallest grammar for $\todisjoint(L^3_n)$ is only linearly larger than the
  smallest grammar for $L^3_n$. By Observation~\ref{obs:disjoint-automata} we
  conclude that the smallest automaton for $L^3_n$ is of the same size as the
  smallest automaton for $\todisjoint(L^3_n)$. Altogether we can conclude that the
  smallest automaton for $\todisjoint(L^3_n)$ is at least exponentially larger than
  the smallest grammar for the same language.

  If Conjecture~\ref{conjecture:UCFG} holds true, we also get exponential lower
  bounds for (E3) and (E4) in an analogous way.
\end{proof}

\subsection{Variable-Length Relations}
Theorems~\ref{theo:size-uniform} and \ref{theo:size-uniform-disjoint} also hold
for variable-length (but finite) relations, as we considered in
Section~\ref{sec:variable-length}. The lower bounds are immediate from those
results and the upper bound constructions are analogous.


\section{Future Work}\label{sec:future}
The connection between database factorization and formal languages gives rise to a plethora of questions for future investigation. What is the complexity of basic operations (e.g., enumeration, counting, and direct access) over compact representations of different formalisms? What is the impact of non-determism on this complexity?  Some questions naturally emerge in the context of answering queries efficiently with compact representations. Given a query and a database, what is the best formal language for representing the result?  What is the impact of the choice of formalism of on our ability to efficiently maintain the query result (as a database view) to accommodate updates in the database?  Specifically, which of the formalisms allow to apply past results on updates of compact representations (e.g.,~\cite{DBLP:conf/pods/AmarilliBMN19,DBLP:conf/icdt/00020OZ23})?  These questions, as well as the connection between factorized relations and path multiset representations, are especially relevant in light of the ongoing efforts on the graph query languages SQL/PGQ and GQL, as these languages combine graph pattern matching (through regular path query evaluation) and relational querying~\cite{DeutschFGHLLLMM-sigmod22,FrancisGGLMMMPR-icdt23}.

In addition to the above, some questions arise independently of the manner (e.g., query) in which the factorized representation is constructed.  What is the complexity of minimizing a representation of each formalism?  What is the tradeoff that the variety of formalisms offers between the size and the complexity of operations? What is the best \emph{lossy} representation if we have a size (space) restriction on the representation? Here, the definition of loss may depend on the application, and one interesting application is \emph{summarization} of a large table via lower-resolution representations, as done by El Gebaly et al.~\cite{DBLP:journals/pvldb/GebalyAGKS14}.


\section*{Acknowledgment}
Part of this work was conducted while Wim Martens was visiting the Simons Institute for the Theory of Computing.

\bibliographystyle{plain}
\bibliography{references}

\begin{toappendix}
\section{The Named and Unnamed Perspectives}\label{app:named}
The principles of databases can be studied in the \emph{named perspective} and
the \emph{unnamed perspective}, which are slightly different mathematical
definitions of the relational model \cite[Chapter 2]{ABLMP21}. Let us provide a bit
of background on both, since the difference between them is important in this
paper.

Let $\values$ and $\attnames$ be two disjoint countably infinite sets of
\emph{values} and \emph{attribute names}. In the \emph{named perspective}, a
\emph{database tuple} is defined as a function $t \colon X \to \values$, where
$X$ is a finite subset of $\attnames$. Such functions $t$ are usually denoted as
$\langle A_1\!:\!a_1, A_2\!:\!a_2, \ldots, A_k\!:\!a_k \rangle$, to say that
$t(A_i) = a_i$ for every $i \in k$. Notice that tuples are unordered in the
sense that $\langle A\!:\!a, B\!:\!b \rangle$ and $\langle B\!:\!b, A\!:\!a
\rangle$ denote the same function. The \emph{arity} of $t$ is $|X|$. A
\emph{database relation} is a set of tuples that are defined over the same set
$X$. A \emph{schema} in the named perspective assigns a finite set of attribute
names to each relation $R$ in the database. The semantics is that each tuple in
$R$ should then be a tuple over the set of attribute names assigned by the
schema.

A \emph{path} in a graph is usually an ordered sequence. (Sometimes containing
only nodes, sometimes only edges, and sometimes nodes and edges depending on the
concrete definition of the graph or multigraph.) Such paths can therefore be
seen as ordered tuples of the form $(u_1,\ldots,u_n)$, where $u_1,\ldots,u_n$
are objects in the graph.

In the \emph{unnamed perspective}, a \emph{($k$-ary) database tuple} is simply
an element of $\V^k$. A \emph{database relation} is a
finite set of tuples of the same arity. We define schemas in the unnamed
perspective a bit differently from the standard definition \cite[Chapter
2]{ABLMP21}, in order to be closer to the named perspective. (In the standard
definition, a schema simply assigns an \emph{arity} to each relation name $R$. The
semantics is that each tuple in $R$ should have the arity that is assigned by
the schema.)
\end{toappendix}

\end{document}